\documentclass[12pt]{article}

\usepackage{amsmath}
\usepackage{amssymb}
\usepackage{amsfonts}
\usepackage{latexsym}
\usepackage{color}

\catcode `\@=11 \@addtoreset{equation}{section}

\catcode `\@=12

\newcommand{\be}{\begin{equation}}
\newcommand{\en}{\end{equation}}
\newcommand{\bea}{\begin{eqnarray}}
\newcommand{\ena}{\end{eqnarray}}
\newcommand{\beano}{\begin{eqnarray*}}
\newcommand{\enano}{\end{eqnarray*}}
\newcommand{\bee}{\begin{enumerate}}
\newcommand{\ene}{\end{enumerate}}

\newcommand{\mc}{\mathcal}

\newcommand{\D}{{\mc D}}

\newcommand{\Sc}{{\cal S}}
\newcommand{\E}{{\cal E}}
\newcommand{\F}{{\cal F}}
\newcommand{\G}{{\cal G}}

\newcommand{\Lc}{{\cal L}}

\newcommand{\1}{1 \!\! 1}

\newcommand{\Hil}{\mc H}

\newtheorem{thm}{Theorem}

\newtheorem{lemma}[thm]{Lemma}
\newtheorem{prop}[thm]{Proposition}
\newtheorem{defn}[thm]{Definition}

\newenvironment{proof}{\noindent {\bf Proof --}}{\hfill$\square$ \vspace{3mm}\endtrivlist}

\catcode `\@=11 \@addtoreset{equation}{section}
\catcode `\@=12

\begin{document}

\thispagestyle{empty}

\vspace*{2cm}

\begin{center}
{\Large \bf More mathematics for pseudo-bosons}   \vspace{2cm}\\

{\large F. Bagarello}\\
  Dipartimento di Energia, Ingegneria dell'Informazione e Modelli Matematici,\\
Facolt\`a di Ingegneria, Universit\`a di Palermo,\\ I-90128  Palermo, Italy\\
e-mail: fabio.bagarello@unipa.it\\
home page: www.unipa.it/fabio.bagarello

\end{center}

\vspace*{2cm}

\begin{abstract}
\noindent We propose an alternative definition for pseudo-bosons. This simplifies the mathematical structure, minimizing the required assumptions. Some physical examples are discussed, as well as some mathematical results related to the biorthogonal sets arising out of our framework.

We also briefly extend the results to the so-called non linear pseudo-bosons.

\end{abstract}

\vspace{2cm}


\vfill


\newpage

\section{Introduction}

In a series of papers, \cite{bagpb1}-\cite{bagrev}, we have considered two operators $a$ and $b$, with $b\neq a^\dagger$, acting on a Hilbert
space $\Hil$, and satisfying, in some suitable sense, the  commutation rule $[a,b]=\1$. A nice functional structure has been deduced under suitable assumptions, and
some connections with physics, and in particular with quasi-hermitian quantum mechanics and with the technique of intertwining operators, have
been established. Following Trifonov, \cite{tri}, we have called {\em pseudo-bosons} (PB) the particle-like excitations associated to this structure. The assumptions used in our construction have been checked for a series of   (quantum mechanical) models. Among other things, we have been forced to introduce a difference between {\em regular} and {\em ordinary} PB. The first ones are those for which, see Section II, the biorthogonal sets of eigenvectors of the operators $N=ba$ and $N^\dagger$, $\F_\varphi$ and $\F_\Psi$, are Riesz bases. On the other hand, when these sets are not Riesz bases, then our PB are not regular.

This paper is motivated by the following, very natural, questions: in the definition of PB we have often required both $\F_\varphi$ and $\F_\Psi$ to be bases for $\Hil$. But, is this really necessary? It is enough, maybe, to require that just one of these two sets is a basis? Or: can we replace this requirement with that of $\F_\varphi$ and/or $\F_\Psi$  being complete? We should recall, in fact, that completeness of a set $\F$ is equivalent to $\F$ being a basis if $\F$ is an orthonormal (o.n.) set, but not in general, at least if $\Hil$ is infinite dimensional, which is the only situation we are interested here in this paper\footnote{A simple reminder: a set $\F_f=\{f_n\in\Hil,\,n\geq0\}$ is a basis if any $h\in\Hil$ admits an unique decomposition in terms of the $f_n$`s. It is complete if zero is the only vector which is orthogonal to all its vectors.}. Actually, there exist intriguingly simple examples of  non o.n. sets, which are complete in $\Hil$ but which are not bases, \cite{you,heil}: let $\E=\{e_n,\, n\geq1\}$ be an o.n. basis for $\Hil$, and let us introduce a new set $\tilde\E:=\{\tilde e_n:=e_n+e_1,\, n=2,3,4,\ldots\}$.
It is clear that $\tilde \E$ is no longer o.n., and it is easy to check that is complete but it is not a basis. Also, its biorthogonal set is easily identified: $\hat\G=\{\hat g_n:=e_n,\,n\geq2\}$, which is not even complete.

Other natural questions are the following: is it, for some reason, automatic that the two sets $\F_\varphi$ and $\F_\Psi$ are complete? Or that they are even bases in $\Hil$?

This is the kind of problems we originally wished to address here. To begin with, it is easy to deduce that the answer to the last two questions is, in general, negative. Indeed, without further assumptions, it is easy to understand that already for ordinary bosonic operators $c$ and $c^\dagger$, with $[c,c^\dagger]=\1$, the set $\chi=\{\chi_n:=\frac{1}{\sqrt{n!}}\,{c^\dagger}^n\chi_0\}$, where $c\chi_0=0$, is not even necessarily complete in $\Hil$. In fact, if $c=\frac{x+ip}{\sqrt{2}}$, with $[x,p]=i\1$, the set $\chi$ is an o.n. basis for $\Hil=\Lc^2(\Bbb R)$ but it is not, for instance, if $\Hil=\Lc^2({\Bbb R}^2)$. In this case, completeness is lost: it is easy, in fact, to find a nonzero function of $\Lc^2({\Bbb R}^2)$ which is orthogonal to all of $\chi_n$. As it is well known, completeness is recovered if we {\em double} the family of ladder operators, that is we consider two operators $c_1$ and $c_2$ satisfying $[c_j,c_k^\dagger]=\delta_{j,k}\1$. This is because $\Lc^2({\Bbb R}^2)$ is isomorphic to $\Lc^2({\Bbb R})\otimes\Lc^2({\Bbb R})$. For this reason, and to avoid these kind of problems, we will fix $\Hil=\Lc^2(\Bbb R)$ in the rest of the paper, where not stated differently, and we will concentrate on this particular situation.

This article is organized as follows: in the next section we propose a different definition for what we call $\D$-PB, that is for those PB which are, somehow, associated to a certain subspace $\D$, dense in the Hilbert space $\Hil$ on which our operators $a$ and $b$ act. This slightly different definition simplifies the treatment of PB quite a bit. In Section III we show how an interesting intertwining relation can be deduced assuming that $a$ and $b$ are related by a third operator, $\Theta$, and we also deduce that the two sets of eigenvectors of the operators $N$ and $N^\dagger$ are related by $\Theta$. In Section IV, after some useful results on biorthogonal sets, we give some physically-motivated examples, while some comments on non linear PB, \cite{bagnlpb1}-\cite{bagzno2}, and our conclusions are discussed in Section V.

\section{A new definition}

We begin this section recalling the definition of linear pseudo-bosons, as originally given in \cite{bagpb1}:

let $\Hil$ be a given Hilbert space with scalar product $\left<.,.\right>$ and related norm $\|.\|$. We introduce a pair of operators, $a$ and
$b$,  acting on $\Hil$ and satisfying the  commutation rule \be [a,b]=\1, \label{21} \en where $\1$ is the identity on $\Hil$.  Of course, this
collapses to the canonical commutation rule (CCR)  if $b=a^\dagger$. Let us call $D^\infty(X):=\cap_{p\geq0}D(X^p)$  the common domain of all the
powers of the operator $X$. In \cite{bagpb1} we have considered the following working assumptions:

\vspace{2mm}

{\bf Assumption 1.--} there exists a non-zero $\varphi_{ 0}\in\Hil$ such that $a\varphi_{ 0}=0$, and $\varphi_{ 0}\in D^\infty(b)$.

{\bf Assumption 2.--} there exists a non-zero $\Psi_{ 0}\in\Hil$ such that $b^\dagger\Psi_{ 0}=0$, and $\Psi_{ 0}\in D^\infty(a^\dagger)$.

{\bf Assumption 3.--}  $\F_\varphi=\{\varphi_{n}=\frac{1}{\sqrt{n!\,}}\,b^{n}\,\varphi_{ 0}\}$ and $\F_\Psi=\{\Psi_{ n}=\,\frac{1}{\sqrt{n!\,}}\,{a^\dagger}^{n}\,\Psi_{ 0}\}$ span the whole $\Hil$.

\vspace{2mm}

We have also considered the following extra assumption, useful but, apparently, not quite physical:

\vspace{2mm}

{\bf Assumption 4.--}  $\F_\Psi$ and $\F_\varphi$ are Riesz bases for $\Hil$.

\vspace{3mm}

For reasons which will appear clear soon, we prefer to consider here a slightly different point of view, which allows us to simplify significantly the procedure. In the present approach the relevant ingredients of our structure will be the two pseudo-bosonic operators $a$ and $b$, {\bf and} a certain dense subset $\D\subset\Hil$, which is stable under the action of $a$, $b$ and of their adjoints. More explicitly, let $a$ and $b$ be two operators on $\Hil$, $a^\dagger$ and $b^\dagger$ their adjoint, and let $\D$ be such that $a^\sharp\D\subseteq\D$ and $b^\sharp\D\subseteq\D$, where $x^\sharp$ is $x$ or $x^\dagger$. Notice that we are not requiring here that $\D$ coincides with, e.g. $D(a)$ or $D(b)$. Of course, $\D\subseteq D(a^\sharp)$ and $\D\subseteq D(b^\sharp)$.

\begin{defn}\label{def21}
The operators $(a,b)$ are $\D$-pseudo bosonic ($\D$-pb) if, for all $f\in\D$, we have
\be
a\,b\,f-b\,a\,f=f.
\label{31}\en
\end{defn}
Due to the stability of $\D$, the above equality is well defined: for instance, since $b\,f\in\D$, it follows that $a$ can safely act on it. Sometimes, to simplify the notation, instead of (\ref{31}) we will simply write $[a,b]=\1$, having in mind that both sides of this equation have to act on $f\in\D$.

It might be interesting to notice that two operators $(a,b)$ which are not $\D_1$-pb, could still be $\D_2$-pb, if $a$, $b$, $\D_1$ and $\D_2$ are chosen properly.

\vspace{2mm}

{\bf Example:--} Let $\Hil=\Lc^2(\Bbb R)$, $a=\frac{d}{dx}$, $b=x$. Let us take $\D_1=\{f(x)\in\Lc^2(\Bbb R):\, f'(x)\in\Lc^2(\Bbb R) \}$. This set is dense in $\Hil$, since it contains the set of the test functions $\Sc(\Bbb R)$, but it is not stable under the action of both $a^\sharp$ and $b^\sharp$. For instance, if $f(x)\in\D_1$, $(bf)(x)=xf(x)$ does not need to belong to $\D_1$ as well. On the other hand, if we take $\D_2=\Sc(\Bbb R)$, this set is stable under $a^\sharp$ and $b^\sharp$. Furthermore, $[a,b]f(x)=f(x)$, for all $f(x)\in \D_2$. Hence $(a,b)$ are $\D_2$-pb, while they are not $\D_1$-pb.

\vspace{2mm}

For these operators the first two assumptions above can be simplified. We now assume that

\vspace{2mm}

{\bf Assumption $\D$-pb 1.--}  there exists a non-zero $\varphi_{ 0}\in\D$ such that $a\varphi_{ 0}=0$.

\vspace{1mm}

{\bf Assumption $\D$-pb 2.--}  there exists a non-zero $\Psi_{ 0}\in\D$ such that $b^\dagger\Psi_{ 0}=0$.

\vspace{2mm}

In fact, if $(a,b)$ satisfy Definition \ref{def21}, it is obvious that $\varphi_0\in D^\infty(b)$ and that $\Psi_0\in D^\infty(a^\dagger)$, so that the vectors
\be
\varphi_n:=\frac{1}{\sqrt{n!}}\,b^n\varphi_0,\qquad \Psi_n:=\frac{1}{\sqrt{n!}}\,{a^\dagger}^n\Psi_0,
\label{32}\en
$n\geq0$, can be defined and they all belong to $\D$. We introduce, as before, $\F_\Psi=\{\Psi_{ n}, \,n\geq0\}$ and
$\F_\varphi=\{\varphi_{ n}, \,n\geq0\}$. Once again, since $\D$ is stable under the action of $a^\sharp$ and $b^\sharp$, we deduce that both $\varphi_n$ and $\Psi_n$ belong to $\D$, so that they belong to the domains of $a^\sharp$, $b^\sharp$ and $N^\sharp$. Now we prove the following

\begin{lemma}\label{lemma21}
The operators $(a,b)$ are $\D$-pb if and only if $(b^\dagger,a^\dagger)$ are $\D$-pb.
\end{lemma}
\begin{proof}
Suppose that $(a,b)$ are $\D$-pb. Then, because of the definition of the adjoint, we can check that $$\left<[b^\dagger,a^\dagger]f,g\right>=\left<f,[a,b]g\right>=\left<f,g\right>,$$
for all $f, g\in\D$. Let now take $\Phi\in\Hil$. Then, since $\D$ is dense in $\Hil$, there exists a sequence $\{\Phi_n\}\subset\D$ converging to $\Phi$. Therefore, recalling that the scalar product is norm continuous, we get
$$
\left<[b^\dagger,a^\dagger]f,\Phi\right>=\lim_n\left<[b^\dagger,a^\dagger]f,\Phi_n\right>=\lim_n\left<f,[a,b]\Phi_n\right>=\lim_n\left<f,\Phi_n\right>=\left<f,\Phi\right>.
$$
Therefore $[b^\dagger,a^\dagger]f=f$ for all $f\in\D$: $(b^\dagger,a^\dagger)$ are $\D$-pb.

The opposite implication can be deduced in a similar way.

\end{proof}
It is now simple to deduce the following lowering and raising relations:
\be
\left\{
    \begin{array}{ll}
b\,\varphi_n=\sqrt{n+1}\varphi_{n+1}, \qquad\qquad\quad\,\, n\geq 0,\\
a\,\varphi_0=0,\quad a\varphi_n=\sqrt{n}\,\varphi_{n-1}, \qquad\,\, n\geq 1,\\
a^\dagger\Psi_n=\sqrt{n+1}\Psi_{n+1}, \qquad\qquad\quad\, n\geq 0,\\
b^\dagger\Psi_0=0,\quad b^\dagger\Psi_n=\sqrt{n}\,\Psi_{n-1}, \qquad n\geq 1,\\
       \end{array}
        \right.
\label{33}\en
as well as the following eigenvalue equations: $N\varphi_n=n\varphi_n$ and $N^\dagger\Psi_n=n\Psi_n$, $n\geq0$, where we recall that $N=ba$ and $N^\dagger=a^\dagger b^\dagger$. In particular, we don't have to bother about the fact that the left-hand sides of these equations are well defined or not, because of what we have already deduced. As a consequence of the eigenvalue equations for $N$ and $N^\dagger$,  choosing the normalization of $\varphi_0$ and $\Psi_0$ in such a way  $\left<\varphi_0,\Psi_0\right>=1$, we deduce that
\be
\left<\varphi_n,\Psi_m\right>=\delta_{n,m},
\label{34}\en
 for all $n, m\geq0$. In fact, since $\left<N\varphi_n,\Psi_m\right>=\left<\varphi_n,N^\dagger\Psi_m\right>$, we have $(n-m)\left<\varphi_n,\Psi_m\right>=0$, which implies that $\left<\varphi_n,\Psi_m\right>=0$ if $n\neq m$. Moreover, the equality $\left<\varphi_n,\Psi_n\right>=1$ can be proved by induction on $n$, using the fact that $\left<\varphi_0,\Psi_0\right>=1$.

So far, no deep difference appears between PB and $\D$-PB. However, it is clear that the stability of $\D$ makes the treatment of these latter much simpler. The main differences arise when considering Assumption 3. The reason is that, in the original definition, we have sometimes implicitly identified completeness of the sets $\F_\varphi$ and $\F_\Psi$ in $\Hil$ with the requirement of they being bases of $\Hil$, at least at the level of the examples\footnote{In fact, in \cite{bagpb4} and \cite{abg}, for instance, we have checked that the sets $\F_\varphi$ and $\F_\Psi$ are complete. This is not enough, see Section IV.}. This is not a problem when the sets are Riesz bases, i.e. when also Assumption 4 above is verified. But, for non regular PB, this is not true in general. We introduce now the following requirement

\vspace{2mm}

{\bf Assumption $\D$-pb 3.--}  $\F_\varphi$ is a basis for $\Hil$.

\vspace{1mm}

This assumption introduces, apparently, an asymmetry between $\F_\varphi$ and $\F_\Psi$, since this last is not required to be a basis as well. Notice also that, if we replace Assumption $\D$-pb 3 with the requirement that $\F_\varphi$ is complete in $\Hil$, the example given in Section I shows that, in general, there is no a priori reason for $\F_\Psi$ to be complete, too. On the other hand, we can prove the following result:

\begin{lemma}
$\F_\varphi$ is a basis for $\Hil$ if and only if $\F_\Psi$ is a basis for $\Hil$.
\end{lemma}
The proof of this statement follows from the uniqueness of the basis biorthogonal to a given basis, \cite{you,heil,chri}. It might be interesting to notice that (i) this lemma reintroduce a complete symmetry between $\F_\varphi$ and $\F_\Psi$, and that (ii) a similar result is false if we simply ask the sets to be complete in $\Hil$, at least for those PB which are not regular. It might be worth also noticing that, while the completeness of $\F_\varphi$ does not imply that $\F_\varphi$ is a basis, the converse is ensured: any basis is complete.

\vspace{2mm}

{\bf Remarks:--} (1) It is interesting to check whether these results can be somehow enriched for our very specific sets $\F_\varphi$ and $\F_\Psi$, which are constructed in a particular way. In fact, this is exactly what happens. We will come back on this aspect later.

(2) If $\F_\varphi$ is a Riesz basis for $\Hil$, we could call our $\D$-PB {\em regular}, as we have done in our previous papers. However, this aspect will not be considered here.

\vspace{2mm}
In view of the examples we will discuss later on, it is also convenient to introduce a weaker form of Assumption $\D$-pb 3: for that we first introduce the notion of $\G$-quasi bases, where $\G$ is a suitable dense subspace of $\Hil$. Two biorthogonal sets $\F_\eta=\{\eta_n\in\G,\,g\geq0\}$ and $\F_\Phi=\{\Phi_n\in\G,\,g\geq0\}$ are {\em $\G$-quasi bases} if, for all $f, g\in \G$, the following holds:
\be
\left<f,g\right>=\sum_{n\geq0}\left<f,\eta_n\right>\left<\Phi_n,g\right>=\sum_{n\geq0}\left<f,\Phi_n\right>\left<\eta_n,g\right>.
\label{iiadd1}
\en
Is is clear that, while Assumption $\D$-pb 3 implies (\ref{iiadd1}), the reverse is false. However, if $\F_\eta$ and $\F_\Phi$ satisfy (\ref{iiadd1}), we still have at hand some (weak) form of resolution of the identity. In fact, formally, we could rewrite (\ref{iiadd1}) as $\sum_{n\geq0}|\eta_n\left>\right<\Phi_n,| = \sum_{n\geq0}|\Phi_n\left>\right<\eta_n,| =\1_\G$. Then our assumption is the following:

\vspace{2mm}

{\bf Assumption $\D$-pbw 3.--}  $\F_\varphi$ and $\F_\Psi$ are $\G$-quasi bases for $\Hil$.

\section{$\D$-conjugate operators}

In this section we slightly refine the structure. Notice that, in what follows, we will always assume that Assumptions $\D$-pb 1, 2 and 3 hold.

We start considering a self-adjoint, invertible, operator $\Theta$, which leaves, together with $\Theta^{-1}$, $\D$ invariant: $\Theta\D\subseteq\D$, $\Theta^{-1}\D\subseteq\D$. Then we introduce the following definition:

\begin{defn}\label{def41}
We will say that $(a,b^\dagger)$ are $\Theta-$conjugate if $af=\Theta^{-1}b^\dagger\,\Theta\,f$, for all $f\in\D$.
\end{defn}
Briefly, we will write $a=\Theta^{-1}b^\dagger\,\Theta$, meaning with that the both sides must be applied to vectors of $\D$. Of course, the fact that $\D$ is stable under the action of both $\Theta$ and $\Theta^{-1}$, makes the above definition well posed, since $\D$ is also stable under the action of $a$ and $b^\dagger$.

Then we have:

\begin{lemma}\label{lemma41} The following statements are all equivalent: 1. $(a,b^\dagger)$ are $\Theta-$conjugate; 2. $(b,a^\dagger)$ are $\Theta-$conjugate; 3. $(a^\dagger,b)$ are $\Theta^{-1}-$conjugate; 4. $(b^\dagger,a)$ are $\Theta^{-1}-$conjugate.
\end{lemma}
\begin{proof}
We just prove here that 1. implies  2. The other statements can be proven in similar way.
Let us assume that $(a,b^\dagger)$ are $\Theta-$conjugate, and let $f, g\in\D$. Then
$$
\left<f,a^\dagger g\right>=\left<af, g\right>=\left<\left(\Theta^{-1}b^\dagger\,\Theta\right)f, g\right>=\left<f, \left(\Theta b\,\Theta^{-1}\right) g\right>,
$$
so that $\left<f,\left(a^\dagger-\left(\Theta b\,\Theta^{-1}\right)\right) g\right>=0$. Then, recalling that the scalar product is continuous and that $\D$ is dense in $\Hil$, we deduce (see the proof of Lemma \ref{lemma21}) that $\left<\hat f,\left(a^\dagger-\left(\Theta b\,\Theta^{-1}\right)\right) g\right>=0$ for all $g\in\D$ and $\hat f\in\Hil$. This implies 2.

\end{proof}

Let us suppose that $\Theta\varphi_0$ is not orthogonal to $\varphi_0$: $\left<\varphi_0,\Theta\varphi_0\right>\neq 0$. We want to show that, if $(a,b^\dagger)$ are $\Theta-$conjugate, then the two sets $\F_\varphi$ and $\F_\Psi$ introduced in the previous section are related by $\Theta$. To prove this, it is convenient to assume that $\left<\varphi_0,\Theta\varphi_0\right>=1$. This is not a major requirement since, if $(a,b^\dagger)$ are $\Theta-$conjugate, then $(a,b^\dagger)$ are also $\hat\Theta-$conjugate, where $\hat\Theta:=\frac{1}{\left<\varphi_0,\Theta\varphi_0\right>}\,\Theta$. With this choice, in fact, $\left<\varphi_0,\hat\Theta\varphi_0\right>=1$. Then we can safely assume the above normalization.
Hence we have:
\begin{prop}\label{prop6}
The operators $(a,b^\dagger)$ are $\Theta-$conjugate if and only if $\Psi_n=\Theta\varphi_n$, for all $n\geq0$.
\end{prop}
\begin{proof}
Let us first assume that $(a,b^\dagger)$ are $\Theta-$conjugate. A simple induction argument shows that $\left<\varphi_n,\Theta\varphi_n\right>=1$ for all $n\geq0$. Indeed this is true for $n=0$. Let us now assume that $\left<\varphi_n,\Theta\varphi_n\right>=1$. Then, using Definition \ref{def41}, the fact that $\varphi_{n+1}=\frac{1}{\sqrt{n+1}}\,b\,\varphi_n$, and the stability of $\D$ under $b^\sharp$ and $\Theta$,
$$
\left<\varphi_{n+1},\Theta\varphi_{n+1}\right>=\frac{1}{n+1}\,\left<\varphi_{n},b^\dagger\Theta\,b\varphi_{n}\right>=\frac{1}{n+1}\,\left<\varphi_{n},\Theta\,a\,b\varphi_{n}\right>=
\frac{1}{n+1}\,\left<\varphi_{n},\Theta\,(N+\1)\varphi_{n}\right>=1,
$$
because of our induction assumption.

The next step consists in proving that $\left<\varphi_{n},\Theta\varphi_{k}\right>=0$ whenever $n\neq k$. This is a standard consequence of the following eigenvalue equation:
$N^\dagger(\Theta\varphi_k)=k(\Theta\varphi_k)$, $\forall \,k\geq0$, which in turn follows from Definition \ref{def41} and Lemma \ref{lemma41}. Hence we conclude that the set $\F_{\tilde\varphi}=\{\tilde\varphi_n:=\Theta\varphi_n,\,n\geq0\}$ is biorthogonal to $\F_\varphi$. To conclude the proof we still have to prove that  $\F_{\tilde\varphi}$ coincides with $\F_{\Psi}$. Indeed, our Assumption $\D$-pb 3 implies that each $f\in\Hil$ can be written as $f=\sum_{k\geq0}\left<\varphi_k,f\right>\Psi_k$. Then, if we take in particular $f\equiv\tilde\varphi_n$, we find that $\tilde\varphi_n=\sum_{k\geq0}\left<\varphi_k,\tilde\varphi_n\right>\Psi_k=\sum_{k\geq0}\delta_{n,k}\Psi_k=\Psi_n$. Hence $\F_{\tilde\varphi}=\F_{\Psi}$\footnote{This is clearly consistent with the existence of an unique basis which is biorthogonal to a given basis, \cite{you}.}.

\vspace{2mm}

Let us now assume that $\Psi_n=\Theta\varphi_n$, for all $n\geq0$. Then, since $a^\dagger$ is a raising operator for $\Psi_n$, $a^\dagger\Psi_n=\sqrt{n+1}\,\Psi_{n+1}$, we deduce that $\Theta^{-1}a^\dagger \Theta\varphi_n=\sqrt{n+1}\,\varphi_{n+1}$, which should be compared with $b\,\varphi_n=\sqrt{n+1}\,\varphi_{n+1}$. Now, let $f$ be a generic vector in $\D$. Then we have
$$
\left<\left(\Theta\,a\,\Theta^{-1}-b^\dagger\right)f,\varphi_{n}\right>=\left<f,\left(\Theta^{-1}\,a^\dagger\,\Theta-b\right)\varphi_{n}\right>=0,
$$
for all $n\geq0$. Hence, since $\F_\varphi$ is complete in $\Hil$, we conclude that $\left(\Theta\,a\,\Theta^{-1}-b^\dagger\right)f=0$ for each $f\in\D$, so that $(b^\dagger,a)$ are $\Theta^{-1}$-conjugate. Our statement follows from Lemma \ref{lemma41}.

\end{proof}

Incidentally we observe that, because of this Proposition, our normalization condition on $\varphi_0$, $\left<\varphi_0,\Theta\varphi_0\right>=1$, can be equivalently stated as a normalization for $\Psi_0$, $\left<\Psi_0,\Theta^{-1}\Psi_0\right>=1$. It is also interesting to stress that, up to this point, we have not required to $\Theta$ to be positive (in some suitable sense). The essential reason is that there is no need for that. In fact,
\begin{prop}
If $(a,b^\dagger)$ are $\Theta-$conjugate then $\left<f,\Theta f\right>>0$ for all non zero $f\in D(\Theta)$.
\end{prop}
\begin{proof}
We first observe that, in general,  the domain of $\Theta$, $D(\Theta)$, is larger than $\D$: $\D\subseteq D(\Theta)\subseteq\Hil$, where $D(\Theta)=\Hil$ only if $\Theta$ is bounded.

Now, each $f\in D(\Theta)$ can be written as $f=\sum_n\left<\Psi_n,f\right>\varphi_n$. Hence, using the continuity of the scalar product, we have
$$
\left<f,\Theta f\right>=\sum_n\left<f,\Psi_n\right>\left<\varphi_n,\Theta f\right>=\sum_n\left<f,\Psi_n\right>\left<\Theta\varphi_n, f\right>=\sum_n\left<f,\Psi_n\right>\left<\Psi_n, f\right>=\sum_n|\left<f,\Psi_n\right>|^2,
$$
which is surely strictly positive if $f\neq0$.

\end{proof}

In some previous paper, \cite{bagzno1,bagzno2}, we have discussed the relation of (non linear) PB with crypto-hermiticity, or its many variations, \cite{mosta}. We are now in the position of repeating a similar analysis in our present settings. In particular, it is a simple exercise to check that, if $(a,b^\dagger)$ are $\Theta-$conjugate, then
\be
Nf=\Theta^{-1}N^\dagger\Theta f,
\label{42}\en
which is our way to say that $N$ is a {\em strongly crypto-hermitian operator}. More in general:
\begin{defn}\label{def31}
Let $X$ be an operator defined on $\D$. We say that $X$ is strongly crypto-hermitian if $Xf=\Theta^{-1}X^\dagger\Theta f$, $\forall \,f\in\D$.
\end{defn}
Notice that, in this definition, we are  fixing two essential ingredients: $\Theta$ and $\D$. Sometimes, if we need to stress these aspects, it might be more convenient to say that $X$ is {\em $(\D,\Theta)$-strongly crypto-hermitian}.

One may wonder wether the previous statement could be inverted: suppose that $N$ is strongly crypto-hermitian. Does it follow that $(a,b^\dagger)$ are $\Theta-$conjugate? In general, the answer seems to be negative, since $a$ and $b^\dagger$ could be, for instance, related as $af=K^{-1}(\Theta^{-1}b^\dagger\Theta)f$, for some $K=K^\dagger$, invertible, mapping $\D$ in $\D$ together with its inverse, and commuting on $\D$ with $\Theta^{-1}b^\dagger\Theta$. Of course, if the only possible choice of an operator $K$ having all these properties is the identity operator, then we could conclude that also the inverse is true. However, we are not yet in a position to get this conclusion. This is work in progress.

\vspace{2mm}

Going back to formula (\ref{42}), we can rewrite it as $\Theta\,Nf=N^\dagger\Theta f$, which shows that $\Theta$ intertwines between $N$ and its adjoint on $\D$. It is easy to check that, choosing, in particular, $f=\varphi_n$, both sides of the equality produce $n\Psi_n$.

We postpone to a future paper the detailed analysis of the consequences of Definition \ref{def31}. This could be particularly interesting,  from a physical point of view, for instance when $X$ is some (generalized) non self-adjoint hamiltonian. We refer to \cite{intop} for some results on intertwining operators.

\section{Some results on biorthogonal sets and some examples}

The examples discussed later in this section will show that the sets $\F_\varphi$ and $\F_\Psi$ share a quite peculiar property: they are related to an o.n. basis via an, in general, unbounded, invertible, operator. This makes the two sets not Riesz bases, for which most of the results which are true for o.n. bases can easily be adapted. On the other hand, see \cite{you,heil,chri}, they are tricky object and some extra care is surely required.

 In the first part of this section we will generalize some of the results holding true for Riesz bases to a slightly more general situation, relevant for those physical applications we will consider later.

Let $\E=\{e_n\in\Hil, n\geq0\}$ be an o.n. basis of $\Hil$ and let us consider a self-adjoint, invertible operator $T$, such that $e_n\in D(T)\cap D(T^{-1})$ for all $n$. Here we are considering the possibility that $T$ or $T^{-1}$, or both, are unbounded. Of course $D(T)$, $D(T^{-1})$ and their intersection $\D$  are, at least, dense in $\Hil$, while they both coincide with $\Hil$ if $T, T^{-1}\in B(\Hil)$. Under our assumption, the vectors  $\varphi_n=Te_n$ and $\Psi_n=T^{-1}e_n$, $n\geq0$, are well defined in $\Hil$. We call $\F_\varphi=\{\varphi_n,\,n\geq0\}$ and $\F_\Psi=\{\Psi_n,\,n\geq0\}$. A simple consequence of these definitions is that $\varphi_n\in D(T^{-1})$, $T^{-1}\varphi_n=e_n$, and $\Psi_n\in D(T)$, $T\Psi_n=e_n$, $n\geq0$. Also, $\Psi_n\in D(T^2)$ and $\varphi_n\in D(T^{-2})$: $T^2\Psi_n=\varphi_n$ and $T^{-2}\varphi_n=\Psi_n$.

We can now prove the following
 \begin{prop}\label{propa1} Under the above assumptions:
 (i) the sets $\F_\varphi$ and $\F_\Psi$ are biorthogonal; (ii) if $f\in D(T)$ is orthogonal to all the $\varphi_n$, then $f=0$; (iii) if $f\in D(T^{-1})$ is orthogonal to all the $\Psi_n$, then $f=0$; (iv) $\forall\,f, g\in\D$ we have $$\left<f,g\right>=\sum_{n=0}^\infty\left<f,\varphi_n\right>\left<\Psi_n,g\right> = \sum_{n=0}^\infty\left<f,\Psi_n\right>\left<\varphi_n,g\right>. $$ Therefore $\F_\varphi$ and $\F_\Psi$ are $\D$-quasi bases; (v) if $T^{-1}$ is bounded, then any $f\in D(T)$ can be written as $f=\sum_{n=0}^\infty\left<\varphi_n,f\right>\Psi_n$. Moreover, if $\hat g\in\Hil$, $\left<f,\hat g\right>=\sum_{n=0}^\infty\left<f,\varphi_n\right>\left<\Psi_n,\hat g\right>$; (vi)  if $T$ is bounded, then any $f\in D(T^{-1})$ can be written as $f=\sum_{n=0}^\infty\left<\Psi_n,f\right>\varphi_n$. Moreover, if $\hat g\in\Hil$, $\left<f,\hat g\right>=\sum_{n=0}^\infty\left<f,\Psi_n\right>\left<\varphi_n,\hat g\right>$.
  \end{prop}
 \begin{proof}
 The proofs of $(i)$, $(ii)$ and $(iii)$ are trivial and will not be given here. To prove $(iv)$ we first observe that if $f, g\in \D$, then both $Tf$ and $T^{-1}g$ are well defined vectors in $\Hil$. Hence, recalling that $\E$ is an o.n. basis and using the definitions of $\varphi_n$ and $\Psi_n$, we get
 $$
 \left<f,g\right>=\left<Tf,T^{-1}g\right>=\sum_{n=0}^\infty\left<Tf,e_n\right>\left<e_n,T^{-1}g\right>=
 \sum_{n=0}^\infty\left<f,\varphi_n\right>\left<\Psi_n,g\right>.
 $$
 Analogously,
 $$
 \left<f,g\right>=\left<T^{-1}f,Tg\right>=\sum_{n=0}^\infty\left<T^{-1}f,e_n\right>\left<e_n,Tg\right>=
 \sum_{n=0}^\infty\left<f,\Psi_n\right>\left<\varphi_n,g\right>.
 $$
 $(v)$ If $f\in D(T)$ we can write $Tf=\sum_{n=0}^\infty\left<e_n,Tf\right>e_n=\sum_{n=0}^\infty\left<\varphi_n,f\right>e_n$. Now
 $$
 \left\|f-\sum_{n=0}^N\left<\varphi_n,f\right>\Psi_n\right\|=\left\|T^{-1}\left(Tf-\sum_{n=0}^N\left<\varphi_n,f\right>e_n\right)\right\|\leq\|T^{-1}\|
 \left\|Tf-\sum_{n=0}^N\left<\varphi_n,f\right>e_n\right\|,
 $$
 which goes to zero when $N$ diverges. The other statement can be proved similarly to $(iv)$.

 $(vi)$ The proof is similar to $(v)$.

\end{proof}

The outcome of this proposition is that we don't really need $\F_\varphi$ and $\F_\Psi$ to be Riesz bases in order to allow a {\em natural} decomposition of most vectors of $\Hil$. This is possible also if one between $T$ and $T^{-1}$ is unbounded, at least if the assumptions under which Proposition \ref{propa1} is stated are satisfied, in some dense subspace of $\Hil$. Of course, when both $T$ and $T^{-1}$ are bounded, then $\F_\varphi$ and $\F_\Psi$ are Riesz bases. However, in the most general case, $\F_\varphi$ and $\F_\Psi$ turn out to be $\D$-quasi bases.






\subsection{Examples}

In some older papers of ours we have considered several examples of PB. We will reconsider few of them, the {\em more physical-motivated ones}, adopting our new point of view.

\subsubsection{The extended quantum harmonic oscillator}

The first example we want to consider was first introduced, in a pseudo-bosonic context, in \cite{bagpb4}. The hamiltonian of this model, introduced in \cite{dapro}, is the  non self-adjoint operator $H_\beta=\frac{\beta}{2}\left(p^2+x^2\right)+i\sqrt{2}\,p$, where $\beta$ is a  strictly positive parameter and $[x,p]=i\1$.

Introducing the standard bosonic operators $a=\frac{1}{\sqrt{2}}\left(x+\frac{d}{dx}\right)$, $a^\dagger=\frac{1}{\sqrt{2}}\left(x-\frac{d}{dx}\right)$, $[a,a^\dagger]=\1$, and the related operators $
A_\beta=a-\frac{1}{\beta}$, and $B_\beta=a^\dagger+\frac{1}{\beta}$,
we have
$
H_\beta=\beta\left(B_\beta A_\beta+\gamma_\beta\,\1\right),
$
where $\gamma_\beta=\frac{2+\beta^2}{2\beta^2}$. It is clear that, for all $\beta>0$, $A_\beta^\dagger\neq B_\beta$ and  $[A_\beta, B_\beta]=\1$. Hence we have to do, apparently, with pseudo-bosonic operators. In \cite{bagpb4} we have deduced, among other results, that the vectors
$$
\varphi_n^{(\beta)}(x)=\frac{1}{\sqrt{n!}}\,B_\beta^n\,\varphi_0^{(\beta)}(x)=\frac{1}{\pi^{1/4}\,\sqrt{2^n\,n!}}\,\left(x-\frac{d}{dx}+\frac{\sqrt{2}}{\beta}\right)^n\,e^{-\frac{1}{2}(x-\sqrt{2}/\beta)^2},
$$
and
$$
\Psi_n^{(\beta)}(x)=\frac{1}{\sqrt{n!}}\,{A_\beta^\dagger}^n\,\Psi_0^{(\beta)}(x)=\frac{1}{\pi^{1/4}\,\sqrt{2^n\,n!}}\,\left(x-\frac{d}{dx}-\frac{\sqrt{2}}{\beta}\right)^n\,
e^{-\frac{1}{2}(x+\sqrt{2}/\beta)^2}.
$$
are eigenstates respectively of $H_\beta$ and $H_\beta^\dagger$ with the same eigenvalue, $\beta(n+\gamma_\beta)$. In particular, the two vacua $\varphi_0^{(\beta)}(x)$ and $\Psi_0^{(\beta)}(x)$ of $A_\beta$ and $B_\beta^\dagger$ are $\varphi_0^{(\beta)}(x)=\frac{1}{\pi^{1/4}}\,e^{-\frac{1}{2}(x-\sqrt{2}/\beta)^2}$ and $\Psi_0^{(\beta)}(x)=\frac{1}{\pi^{1/4}}\,e^{-\frac{1}{2}(x+\sqrt{2}/\beta)^2}$. Also, we have shown that the operator $V_\beta=e^{(a+a^\dagger)/\beta}=e^{\sqrt{2}x/\beta}$, together with its inverse, map the o.n. basis $\E:=\{e_n(x)=\frac{1}{\sqrt{n!}}(a^\dagger)^ne_0(x), \,n\geq0\}$, where $a\,e_0(x)=0$, into  $\F_\varphi^{(\beta)}=\{\varphi_n^{(\beta)}(x),\,n\geq 0\}$ and $\F_\Psi^{(\beta)}=\{\Psi_n^{(\beta)}(x),\,n\geq 0\}$, respectively. More exactly, $\varphi_n^{(\beta)}=e^{-1/\beta^2}V_\beta\varphi_n$, and $\Psi_n^{(\beta)}=e^{1/\beta^2}V_\beta^{-1}\varphi_n$, $n\geq0$. This suggests to identify the operator $T$ of Propositions \ref{propa1} with $e^{-1/\beta^2}V_\beta$. It is clear that
$$
D(T)=\{f(x)\in\Lc^2(\Bbb R):\, e^{\sqrt{2}x/\beta}f(x)\in \Lc^2(\Bbb R)\},$$  and $$D(T^{-1})=\{f(x)\in\Lc^2(\Bbb R):\, e^{-\sqrt{2}x/\beta}f(x)\in \Lc^2(\Bbb R)\}.
$$
These sets are dense in $\Lc^2(\Bbb R)$, since both contain the set $\Sc(\Bbb R)$ of fast decreasing functions. Then, since $e_n(x)\in\Sc(\Bbb R)$,  Proposition \ref{propa1} holds and we conclude that $\F_\varphi^{(\beta)}$ and $\F_\Psi^{(\beta)}$ are biorthogonal $D(T)\cap D(T^{-1})$-quasi bases, as  required by Assumption $\D$-pbw 3. We also recall that, \cite{bagpb4}, they are both complete in $\Lc^2(\Bbb R)$.

Concerning Assumptions $\D$-pb 1 and 2, a set $\D$ with the required properties does exist: we take $\D\equiv \Sc(\Bbb R)$. It is clear that both $\varphi_0^{(\beta)}(x)$ and $\Psi_0^{(\beta)}(x)$ belong to $\Sc(\Bbb R)$, and that $A_\beta^\sharp$ and $B_\beta^\sharp$ leave this space stable. Needless to say, $\Sc(\Bbb R)$ is also dense in $\Lc^2(\Bbb R)$. Therefore, most requirements discussed in Section II are satisfied.

\subsubsection{The Swanson model}

The starting point is the non self-adjoint hamiltonian,
$
H_\theta=\frac{1}{2}\left(p^2+x^2\right)-\frac{i}{2}\,\tan(2\theta)\left(p^2-x^2\right),
$
where $\theta$ is a real parameter taking value in $\left(-\frac{\pi}{4},\frac{\pi}{4}\right)\setminus\{0\}=:I$, \cite{dapro}.  As before, $[x,p]=i\1$. Introducing the annihilation and creation operators $a$, $a^\dagger$, and their linear combinations $A_\theta=\cos(\theta)\,a+i\sin(\theta)\,a^\dagger=\frac{1}{\sqrt{2}}\left(e^{i\theta}x+e^{-i\theta}\,\frac{d}{dx}\right)$ and $B_\theta=\cos(\theta)\,a^\dagger+i\sin(\theta)\,a=\frac{1}{\sqrt{2}}\left(e^{i\theta}x-e^{-i\theta}\,\frac{d}{dx}\right)$, we can write
$
H_\theta=\omega_\theta\left(B_\theta\,A_\theta+\frac{1}{2}\1\right),
$
where $\omega_\theta=\frac{1}{\cos(2\theta)}$ is well defined because $\cos(2\theta)\neq0$ for all $\theta\in I$. It is clear that $A_\theta^\dagger\neq B_\theta$ and that $[A_\theta,B_\theta]=\1$. The two vacua of $A_\theta$ and $B_\theta^\dagger$ are $
\varphi_0^{(\theta)}(x)=N_1 \exp\left\{-\frac{1}{2}\,e^{2i\theta}\,x^2\right\},
$
 and $
\Psi_0^{(\theta)}(x)=N_2 \exp\left\{-\frac{1}{2}\,e^{-2i\theta}\,x^2\right\},
$
where $N_1$ and $N_2$ are suitable normalization constants. Notice that, since $\Re(e^{\pm 2i\theta})=\cos(2\theta)>0$ for all $\theta\in I$, both $\varphi_0^{(\theta)}(x)$ and $\Psi_0^{(\theta)}(x)$ belong to $\Lc^2({\Bbb R})$. The functions of the sets $\F_\varphi^{(\theta)}$ and $\F_\Psi^{(\theta)}$ are found in \cite{bagpb4}:
$$
\left\{
\begin{array}{ll}
\varphi_n^{(\theta)}(x)=\frac{N_1}{\sqrt{2^n\,n!}}
\,H_n\left(e^{i\theta}x\right)\,\exp\left\{-\frac{1}{2}\,e^{2i\theta}\,x^2\right\},  \\
\Psi_n^{(\theta)}(x)=\frac{N_2}{\sqrt{2^n\,n!}}
\,H_n\left(e^{-i\theta}x\right)\,\exp\left\{-\frac{1}{2}\,e^{-2i\theta}\,x^2\right\},
\end{array}
\right.
$$
where $H_n(x)$ is the n-th Hermite polynomial. Furthermore, in \cite{bagpb4} we have also deduced that a non zero complex constant $\alpha$ does exist such that $\varphi_n^{(\theta)}(x)=\alpha\, T_\theta\, e_n(x)$,  and  $\Psi_n^{(\theta)}(x)=\frac{1}{\overline\alpha}\, T_\theta^{-1}\, e_n(x)$, for all $n\geq0$, where the $e_n(x)$'s are the same as in the previous example, and $T_\theta=e^{i\frac{\theta}{2}(a^2-{a^\dagger}^2)}=e^{i\frac{\theta}{2}\left(x\frac{d}{dx}+\frac{d}{dx}x\right)}$ is a self-adjoint, invertible, unbounded operator. From now on, to simplify the notation, we will assume $\alpha=1$. Since $(T_\theta f)(x)=e^{i\frac{\theta}{2}}f(e^{i\theta}x)$ (for all functions for which this formula makes sense), it is clear that
$$
D(T_\theta)=\{f(x)\in\Lc^2(\Bbb R):\, f(e^{i\theta}x)\in \Lc^2(\Bbb R), \,\forall\theta\in I\},$$ $$D(T_\theta^{-1})=\{f(x)\in\Lc^2(\Bbb R):\, f(e^{-i\theta}x)\in \Lc^2(\Bbb R), \,\forall\theta\in I\}.
$$
They are both dense, together with their intersection $D(T_\theta)\cap D(T_\theta^{-1})$, in $\Lc^2(\Bbb R)$, since all these sets contain the linear span of the $e_n(x)$'s, $\Lc_\E$: each finite linear combination of the $e_n(x)=\frac{1}{\sqrt{2^n\,n!}}H_n(x)e^{-\frac{1}{2}\,x^2}$ clearly belongs to both $D(T_\theta)$ and $D(T_\theta^{-1})$. Notice that this is true because $\theta\in I$. Otherwise the statement would be false. Obviously, $\Lc_\E$ is dense in $\Lc^2(\Bbb R)$, since $\E$ is an o.n. basis for $\Lc^2(\Bbb R)$. Now, our Proposition \ref{propa1} can be applied and the conclusion is that $\F_\varphi^{(\theta)}$ and $\F_\Psi^{(\theta)}$ are $D(T_\theta)\cap D(T_\theta^{-1})$-quasi bases. We also recall that they are both complete in $\Lc^2(\Bbb R)$, \cite{bagpb4}.

The space $\D$ is, as in the previous example, $\Sc(\Bbb R)$. This is stable under the action of $A_\theta^\sharp$ and $B_\theta^\sharp$ , and $\Psi_0^{(\theta)}(x), \varphi_0^{(\theta)}(x)\in \Sc(\Bbb R)$.

\subsubsection{Generalized Landau levels}

The details of the model are discussed in \cite{abg}, and it is surely not worth, and too long, repeating them here. However, we need to stress that this example is a two-dimensional version of what we have discussed so far in this paper.

 The essential idea is that we have a non self-adjoint hamiltonian acting on $\Hil=\Lc^2({\Bbb R}^2)$, which with a suitable choice of variables, can be written as $h'=B'A'-\frac{1}{2}\1$, where $A'=\alpha'\left(\partial_x-i\partial_y+\frac{x}{2}(1-2k_2)-\frac{iy}{2}(1-2k_1)\right)$ and $B'=\gamma'\left(-\partial_x-i\partial_y+\frac{x}{2}(1-2k_2)+\frac{iy}{2}(1+2k_1)\right)$, for suitable complex constants $\alpha'$ and $\gamma'$, and for $k_j\in\left]-\frac{1}{2},\frac{1}{2}\right[$. This hamiltonian commutes with a second, again non self-adjoint, operator $h=B\,A-\frac{1}{2}\1$, a second hamiltonian, where $A=\alpha\left(-i\partial_x+\partial_y-\frac{ix}{2}(1+2k_2)+\frac{y}{2}(1-2k_1)\right)$ and $B=\gamma\left(-i\partial_x-\partial_y+\frac{ix}{2}(1-2k_2)+\frac{y}{2}(1+2k_1)\right)$, $\alpha, \gamma\in \Bbb{C}$ chosen properly, \cite{abg}. We have discussed in \cite{abg} in which sense this model extends the ordinary two-dimensional hamiltonian of the Landau levels to a non self-adjoint situation. In particular, we go back to Landau levels simply taking $k_1=k_2=0$ and $\alpha= \alpha'= \gamma= \gamma'=\frac{1}{\sqrt{2}}$.

The vacua of $A$, $A'$ and of $B^\dagger$, $B'^\dagger$ are found to be   $\varphi_{0,0}(x,y)=N_\varphi\,e^{\left\{-\frac{x^2}{4}(1+2k_2)-\frac{y^2}{4}(1-2k_1)\right\}}$ and
$\Psi_{0,0}(x,y)=N_\Psi\,e^{\left\{-\frac{x^2}{4}(1-2k_2)-\frac{y^2}{4}(1+2k_1)\right\}},$
 where $N_\varphi$ and $N_\Psi$ are normalization
constants  chosen in such a way that
$\left<\varphi_{0,0},\Psi_{0,0}\right>=1$.

In \cite{abg} it is shown that the vectors $$\varphi_{n,l}(x,y)=\frac{B'^n\,B^l}{\sqrt{n!\,l!}}\,\varphi_{0,0}(x,y), \mbox{ and }
\Psi_{n,l}(x,y)=\frac{(A'^\dagger)^n\,(A^\dagger)^l}{\sqrt{n!\,l!}}\,\Psi_{0,0}(x,y),$$
 $n,l=0,1,2,3,\ldots$, are related to the vectors $e_{n,l}(x,y)$ of an o.n. basis of $\Lc^2({\Bbb R}^2)$, $\E=\{e_{n,l}(x,y),\,n,l\geq0\}$, in a simple way. Here, see \cite{abg}, $e_{n,l}(x,y)=\frac{1}{\sqrt{2^{n+l}\,n!\,l!}}\, H_n(x)H_l(y)e^{-\frac{1}{4}(x^2+y^2)}$ produces the o.n. basis of a two-dimensional harmonic oscillator. In particular, we have shown that $\varphi_{n,l}(x,y)=T e_{n,l}(x,y)$, while $\Psi_{n,l}(x,y)=T^{-1} e_{n,l}(x,y)$, $n, l\geq0$, with $T=\sqrt{2\pi} N_\varphi e^{-\frac{x^2}{2}\,k_2+\frac{y^2}{2}k_1}$, a simple multiplication operator. We now have
 $$
 D(T)=\left\{f(x,y)\in\Lc^2({\Bbb R}^2): e^{-\frac{x^2}{2}\,k_2+\frac{y^2}{2}k_1}f(x,y)\in\Lc^2({\Bbb R}^2)\right\},
 $$
 and a similar definition can be deduced for $D(T^{-1})$. We stress again that $k_j\in\left]-\frac{1}{2},\frac{1}{2}\right[$ here. These two sets are dense in $\Lc^2({\Bbb R}^2)$, together with their intersection, since they both contain $\Lc_\E$, the linear span of the $ e_{n,l}(x,y)$'s. Then, Proposition \ref{propa1} implies that $\F_\varphi=\{\varphi_{n,l}(x,y)\}$ and $\F_\Psi=\{\Psi_{n,l}(x,y)\}$ are $D(T)\cap D(T^{-1})$-quasi bases for $\Hil$. They are also complete in  $\Lc^2({\Bbb R}^2)$, \cite{abg}.

Concerning Assumptions $\D$-pb 1 and 2, a set $\D$ with the required properties does exist: for that we take $\D\equiv \Sc({\Bbb R}^2)$. It is clear that both $\varphi_{0,0}(x,y)$ and $\Psi_{0,0}(x,y)$ belong to $\Sc({\Bbb R}^2)$, and that $A^\sharp$, $A'^\sharp$,  $B^\sharp$ and $B'^\sharp$, all leave this space stable. Needless to say, $\Sc({\Bbb R}^2)$ is also dense in $\Lc^2({\Bbb R}^2)$. Therefore, all the requirements discussed in Section II are satisfied.

\vspace{3mm}

{\bf Remarks:--} (1) It turns out from our general results and from our examples here that the set $D(T)\cap D(T^{-1})$ is simply a {\em supplementary space}, useful to investigate the nature of $\F_\varphi$ and $\F_\Psi$ but not strictly related, in principle, to the $\D$-pb nature of the operators $a$ and $b$. On the other hand, the role of $\D$ is crucial to keep the mathematics of the procedure simple and {\em under control}.

(2) More details on the above examples, as well as other examples, can be found in \cite{bagpb4,abg}.

\section{$\D$-non linear PB and conclusions}

The new definition of PB proposed here can be easily extended to what we have called {\em non linear} PB, \cite{bagnlpb1}-\cite{bagzno2}. Let us consider a strictly increasing sequence $\{\epsilon_n\}$: $0=\epsilon_0<\epsilon_1<\cdots<\epsilon_n<\cdots$. Then, given two operators $a$ and $b$ on $\Hil$, and a set $\D\subset\Hil$ which is dense in $\Hil$,

\begin{defn}
We will say that the triple $(a,b,\{\epsilon_n\})$ is a family of $\D$-non linear pseudo-bosons ($\D$-NLPB) if the following properties hold:
\begin{itemize}

\item {\bf p1.} a non zero vector $\Phi_0$ exists in $\D$ such that $a\,\Phi_0=0$;

\item {\bf  p2.} a non zero vector $\eta_0$ exists in $\D$ such that $b^\dagger\,\eta_0=0$;

\item {\bf { p3}.} Calling
\be
\Phi_n:=\frac{1}{\sqrt{\epsilon_n!}}\,b^n\,\Phi_0,\qquad \eta_n:=\frac{1}{\sqrt{\epsilon_n!}}\,{a^\dagger}^n\,\eta_0,
\label{55}
\en
we have, for all $n\geq0$,
\be
a\,\Phi_n=\sqrt{\epsilon_n}\,\Phi_{n-1},\qquad b^\dagger\eta_n=\sqrt{\epsilon_n}\,\eta_{n-1}.
\label{56}\en
\item {\bf { p4}.} The set $\F_\Phi=\{\Phi_n,\,n\geq0\}$ is a basis for $\Hil$.

\end{itemize}

\end{defn}

{\bf Remarks:--} (1) Since $\D$ is stable under the action of $b$ and $a^\dagger$, it follows that $\Phi_n, \eta_n\in \D$, for all $n\geq0$.

(2) $\D$-PB are recovered choosing $\epsilon_n=n$.

(3) If $\F_\Phi$ is a Riesz basis for $\Hil$, the $\D$-NLPB are called {\em regular}, in agreement with our previous notation.

(4) The set $\F_\eta=\{\eta_n,\,n\geq0\}$ is automatically a basis for $\Hil$ as well. This follows from the fact that, calling $M=ba$, we have $M\Phi_n=\epsilon_n\Phi_n$ and $M^\dagger\eta_n=\epsilon_n\eta_n$. Therefore, choosing the normalization of $\eta_0$ and $\Phi_0$ in such a way $\left<\eta_0,\Phi_0\right>=1$, $\F_\eta$ is biorthogonal to the basis $\F_\Phi$. Then, it is possible to check that $\F_\eta$ is the unique basis which is biorthogonal to $\F_\Phi$.

(5) It could be useful to introduce, in the present context, the notion of $\G$-quasi bases. However, this will not be done here.

\vspace{2mm}

Also in this context it is possible to deduce interesting intertwining relations. We just consider here the simple situation, motivated by Proposition \ref{prop6}, in which the two bases are related by a suitable self-adjoint, invertible and, in general, unbounded operator $\Theta$ which, together with $\Theta^{-1}$, leaves $\D$ invariant. More explicitly, we require that $\eta_n=\Theta\Phi_n$, $\forall\,n$. In this case we easily get $$\left(M^\dagger\Theta-\Theta M\right)\Phi_n=0,$$
for all $n$. Therefore, $M$ is strongly crypto-hermitian. As already stressed in \cite{bagnlpb1}-\cite{bagzno2}, this could be relevant in discussing physical systems described by some hamiltonian which is not self-adjoint, but crypto-hermitian or PT-symmetric, and with eigenvalues $\epsilon_n$ which are not necessarily linear in the quantum number $n$.

\vspace{2mm}

We have proposed a slightly {\em improved} definition of PB, which we have called $\D$-PB, for which the same original results deduced for ordinary PB can be deduced in a simpler way. Adopting this definition, some of the original assumptions can also be weakened, making, in our opinion, all the construction rather elegant. We have also discussed some explicit quantum mechanical example and we have extended our construction to the non linear case.

\section*{Acknowledgements}

The author acknowledges financial support by the MIUR. He also thanks very much Prof. Christopher Heil, Prof. Camillo Trapani and Dr. Petr Siegl for their precious advices during the preparation of this paper.

\end{document}